\documentclass[12pt,reqno]{amsart}
\usepackage{amssymb}
\usepackage[russian,english]{babel}
\usepackage[dvipsnames]{xcolor}
\usepackage{comment}
\usepackage{graphicx}
\usepackage{hyperref}
\usepackage[utf8]{inputenc}
\usepackage{mathabx}
\usepackage[vskip=8pt]{quoting}
\usepackage{setspace}
\usepackage{url}

\newtheorem{theorem}{Theorem}
\newtheorem{claim}[theorem]{Claim}
\newtheorem{corollary}[theorem]{Corollary}
\newtheorem{lemma}[theorem]{Lemma}
\newtheorem{proposition}[theorem]{Proposition}

\theoremstyle{definition}
\newtheorem{definition}[theorem]{Definition}

\newenvironment{YG}{\noindent\color{Maroon}}{}

\newcommand\Ar{\smallskip\noindent\textbf{A: }}
\newcommand\Qn{\smallskip\noindent\textbf{Q: }}

\newcommand{\bra}[1]{\ensuremath{\langle#1|}}
\newcommand{\ket}[1]{\ensuremath{|#1\rangle}}
\newcommand\braket[2]{\ensuremath{\langle#1\,|\,#2\rangle}}
  
  \newcommand\Bigbraket[2]
    {\ensuremath{\Big\langle\,#1\:\Big|\:#2\,\Big\rangle}}

\newcommand\Co{\ensuremath{\mathbb C}}
\renewcommand\d{\ensuremath{\delta}}
\newcommand\dg{^\dag}
\newcommand\F{\ensuremath{\mathcal F}}

\renewcommand\H{\ensuremath{\mathcal H}}
\newcommand\h{\ensuremath{\hbar}}
\newcommand\iset[1]{\ensuremath{\left\langle#1\right\rangle}}
\renewcommand\l{\ensuremath{\lambda}}
\newcommand\ltr{\ensuremath{L^2(\R)}}
\newcommand\mo{^{-1}} 
\renewcommand\P{\ensuremath{\mathcal P}}
\newcommand\pa[1]{\ensuremath{\left(#1\right)}} 
\renewcommand{\phi}{\ensuremath{\varphi}}
\newcommand\qef{\hfill$\triangleleft$} 
\newcommand\R{\ensuremath{\mathbb R}}
\newcommand\set[1]{\ensuremath{\left\{#1\right\}}}
\newcommand\vp[1]{\ensuremath{\left|#1\right|}} 
\newcommand\x{\ensuremath{\times}}

\title[Wigner's quasidistribution]{Wigner's quasidistribution\\ and Dirac's kets}
\author[Blass]{Andreas Blass}
\address{Mathematics Department, University of Michigan, USA}
\email{ablass@umich.edu}
\author[Gurevich]{Yuri Gurevich}
\address{Computer Science and Engineering,
University of Michigan, USA}
\email{gurevich@umich.edu}
\author[Volberg]{Alexander Volberg}
\address{Mathematics Department,
Michigan State University, USA, and
Hausdorff Center for Mathematics, University of Bonn, Germany}
\email{volberg@msu.edu}

\begin{document}

\begin{abstract}
In every state of a quantum particle, Wigner's quasidistribution is the unique quasidistribution on the phase space with the correct marginal distributions for position, momentum, and all their linear combinations.
\end{abstract}

\maketitle
\thispagestyle{empty}

\footnotetext{Partially supported by the US Army Research Office under W911NF-20-1-0297 (authors 1 and 2) and the NSF grant DMS
1900286 (author 3).}

\bigskip
\mbox{}\hfill
\begin{minipage}{\textwidth}
\begin{quoting}\raggedleft\small\it
The only difference between a probabilistic classical
world and the equations of the quantum world is that somehow or other it appears as if the probabilities would have to go negative \dots
Okay, that's the fundamental problem.
I don't know the answer to it, \dots
if I try my best to make the equations look as near as possible to what would be imitable by a classical probabilistic computer, I get into trouble.\\[10pt]
--- Richard Feynman,\\
Simulating Physics with Computers, 1982 \cite[p.~480]{Feynman}
\end{quoting}
\end{minipage}
\bigskip\bigskip

\section{Introduction}
\label{sec:intro}

The story of negative probabilities starts with the 1932 article \cite{Wigner} by Eugene Wigner.
In quantum mechanics, probability distributions of the position and momentum of a particle make physical sense but their joint distribution doesn't. Yet Wigner exhibited such a joint distribution.
It had some desired properties.
However, some of its values were negative.
``But of course,'' wrote Wigner, ``this must not hinder the use of it in calculations.''

In 1987, Jacqueline and Pierre Bertrand proposed ``a new derivation of Wigner's function based on the property of positivity of its integrals along straight lines in phase space'' \cite{Bertrand}.
In 2014, in this Bulletin \cite{G224}, we sketched a mathematical proof of a characterization of Wigner’s quasidistribution as the unique quasidistribution on the phase space $\R^2$ that yields the correct marginal distributions not only for position and momentum but for all their linear combinations.
In 2021, that sketch was developed into a complete proof that the characterization is valid in ``nice'' states, namely the states given by smooth functions with compact support \cite{G245}.

In this paper, we prove that the characterization is valid in all states, with no exception.
Furthermore, the new proof is simpler, conceptually and technically.
In particular, the uniqueness is derived from a purely measure-theoretic observation that we prove in \S\ref{sec:pushforward}.
The simplicity of the new proof gave us the idea to present it in this Bulletin.
We made an effort to give our readers a comprehensible and maybe even enjoyable introduction to some foundational issues of science.

\bigskip\noindent\textbf{Quisani%
\footnote{A former student of the second author.}: }
What are negative probabilities?

\smallskip\noindent\textbf{Authors: }
The axiomatic definition of probabilities readily generalizes to \emph{quasiprobabilities}, or \emph{signed probabilities}, where negative values are allowed \cite{G245}.
Basically, you drop the requirement that probabilities take values in the real segment $[0,1]$ and allow arbitrary real values.

\Qn But what is the intuition behind negative probabilities?
An urn cannot have $-3$ red balls.

\Ar We don't know.
At this point, we find it more fruitful to think about what quasiprobabilities are good for.

\Qn You once compared the generalization of probabilities to quasiprobabilities with the generalization of real numbers to complex ones.
Complex numbers became indispensable, e.g., in solving algebraic equations.
Are there important theoretical problems that have been solved using negative probabilities?

\Ar We don't know such problems, but we expect that quasiprobabilities will be used to solve theoretical problems.
They are already used in practice.

\Qn Yes, you mentioned quantum tomography in our 2014 conversation \cite{G224}.
Being a software engineer, I realize the paramount value of practical applications.
But today I would like you to address basic questions.
Some of these basic questions you seemed to dodge during our 2014 conversation.
For example, you spoke about marginal distributions not only for the position and momentum but also for all their linear combinations.
But what are marginal distributions for linear combinations?
You never defined them properly.

\Ar Addressing basic questions is fine, and we will define those marginal distributions.
We will try to explain things the best we can.

\Qn Do explain. But please take into account that, in the meantime, I was busy with computer engineering.
I didn't have time to study quantum mechanics or measure theory.

\Ar Understood.

\section{Preliminaries}
\label{sec:prelim}

\subsection{Measures}
We recall some basic definitions of measure theory.

A \emph{measurable space} $M$ is a pair $(\Omega,\Sigma)$ where $\Omega$ is a nonempty set and $\Sigma$ a $\sigma$-algebra of subsets of $\Omega$.
In other words, $\Sigma$ is a Boolean algebra closed under countable unions.
Members of $\Sigma$ are \emph{measurable sets} of $M$.

\noindent{\tt Example:}
The real line \R\ with the collection of real Borel sets, which  is the least $\sigma$-algebra containing every open real interval $(a,b)$. \qef

\noindent{\tt Example:}
The real plane $\R^2$ with the collection of Borel subsets of $\R^2$, which is the least $\sigma$-algebra containing every open rectangle $(a,b) \times (c,d)$. \qef

A \emph{measure} $\mu$ on a nonempty set $\Omega$ is a function such that
\begin{enumerate}
\item the domain of $\mu$ is a $\sigma$-algebra of subsets of $\Omega$ known as \emph{$\mu$-measurable} sets,
\item $\mu$ assigns a real number or $\infty$ to each $\mu$-measurable set, and
\item $\mu$ is \emph{countably additive} which means that, for all pairwise disjoint measurable sets $s_n$, we have
\[ \mu\Big(\bigcup_{n=1}^\infty s_n\Big) = \sum_{n=1}^\infty \mu(s_n). \]
Since the union $\bigcup s_n$ is independent of the order of the sets $s_n$, so is the sum $\sum \mu(s_n)$. That implies absolute convergence by a well-known theorem of Riemann.
\end{enumerate}
If $\mu$ does not take value $\infty$, then $\mu$ is \emph{finite}.
If $\mu$ has no negative values then it is \emph{nonnegative}.
Every measure we consider in this paper is either finite or nonnegative.

\noindent{\tt Example:}
The Lebesgue measure on Euclidean spaces $\R^k$ and finite-dimensional Hilbert spaces $\Co^k$, called \emph{length} in the case of \R, called \emph{area} in the cases of $\R^2$ and \Co, and called \emph{volume} in the case of $\R^3$ and in general.
A careful treatment of the Lebesgue measure, with all the necessary proofs, is somewhat involved \cite[Chapter~11]{Rudin}, but the definition itself is simple, and we give a version of it on the example of the interval $(0,1)$ in \R.

An open set $O$ in $(0,1)$ is the disjoint union of its maximal intervals; define the length of $O$ to be the sum of the lengths of its maximal intervals.
For any set $s\subseteq (0,1)$, the \emph{outer measure} $\lambda^*(s)$ of $s$ is the infimum of the lengths of the open sets $O$ that cover $s$.
The \emph{inner measure} $\lambda_*(s)$ is defined as $1 - \lambda^*((0,1) - s)$.

If $\lambda^*(s) = \lambda_*(s)$, then  $s$ is \emph{Lebesgue measurable} and the outer (and also the inner) measure $\lambda^*(s)$ is called the \emph{Lebesgue measure} $\lambda(s)$.

The Lebesgue measure on \R\ is defined by applying this construction to intervals $(i,i+1)$ for all integers $i$ and adding the resulting measures if all the pieces are measurable.

\begin{lemma}\label{lem:diff}
If $\mu, \nu$ are finite measures on a measurable space $M$, then their difference
$ (\mu-\nu)(s) = \mu(s) - \nu(s) $
is a finite measure on $M$.
\end{lemma}

\begin{proof}
If measurable sets $s_1, s_2, \dots$ are pairwise disjoint, then
\begin{align*}
(\mu-\nu)\bigcup_n s_n &= \mu\bigcup_n s_n - \nu\bigcup_n s_n= \sum_n \mu(s_n) - \sum_n \nu(s_n)\\
&= \sum_n (\mu(s_n) - \nu(s_n))
= \sum_n (\mu-\nu)(s_n).
\end{align*}
The third equality holds because of absolute convergence.
\end{proof}

A function $f: \Omega_1\to\Omega_2$ from a measurable space $M_1 = (\Omega_1,\Sigma_1)$ to a measurable space $M_2 = (\Omega_2,\Sigma_2)$ is \emph{measurable for $M_1, M_2$} if the $f$-preimage of every measurable set in $M_2$ is measurable in $M_1$.
If $M_2$ is the real line \R\ or complex line \Co\ endowed with the $\sigma$-algebra of Borel sets, then $f$ is a \emph{measurable function on $M_1$}.
\qef

\subsection{\ltr, and test functions}

If $f$ is a measurable function on a Euclidean space $\R^k$ and $\mu$ is a nonnegative measure on $\R^k$, then
\[ \int f d\mu = \int_{\R^k} f d\mu = \int_{\R^k} f(x) d\mu(x) \]
means the Lebesgue integral of $f$ with respect to measure $\mu$.
For real-valued $f$, the integral is defined by approximating $f$ with so-called \emph{simple functions}, i.e.\ functions $g(x)$ taking only finitely many values $v_i$, each on a measurable set $s_i$.
The integral $\int g(x) d\mu(x)$ is simply $\sum_i v_i\mu(s_i)$.
If the supremum of the integrals of simple functions $g(x)\le f(x)$
for all $x$ coincides with the infimum of the integrals of simple
functions $g(x)\ge f(x)$ for all $x$, then their common value is the
integral $\int f d\mu$, in which case $f$ is \emph{integrable} with
respect to measure $\mu$.
See details in Chapter~11 of \cite{Rudin}.

For \Co-valued functions $f$, just integrate the real and imaginary parts separately.
It is easy to check that every bounded continuous function is integrable with respect to any finite measure $\mu$.

If $s$ is a measurable subset of $\R^k$, we let $\chi_s(x)$ be 1 for $x\in s$ and 0 otherwise.
Then $\int_s f(x)\, d\mu(x)$ means $\int_{\R^k} \chi_s(x)f(x)\, d\mu(x)$.

\noindent{\tt Proviso.}
By default, Euclidean spaces $\R^k$ come with the Lebesgue measure. \qef

\ltr\ is the Hilbert space of square integrable functions $\psi: \R \to\Co$ with the inner product \braket\psi\phi\
given by  the (Lebesgue) integral
$\displaystyle \int_\R \psi^*(x)\phi(x)\, dx$.

Two \ltr\ functions are considered equivalent if they differ only on a set of measure zero.
Strictly speaking, $L^2(\R)$ vectors are the equivalence classes.
It is more convenient though to work with individual functions modulo the equivalence relation.

The forward Fourier transform \F\ sends an \ltr\ function $\psi(x)$ to
\[
 \widehat{\psi}(\xi)  =
    \frac1{\sqrt{2\pi}}\int \psi(x)\,e^{-i\xi x}\,dx
\]
provided that the integral exists.
Similarly, the inverse Fourier transform $\F\mo$ sends a function $\phi(\xi)$ to
\[
 \widecheck{\phi}(x) =
    \frac1{\sqrt{2\pi}} \int \phi(\xi)e^{i\xi x}\,d\xi,\\
\]
Mathematically $x$ and $\xi$ are real variables. In applications, the dimension of $\xi$ is the inverse of that of $x$ so that $\xi x$ is a pure number.
Here and in the rest of the paper, integrals are by default integrals over \R.

The forward and inverse Fourier transforms are defined also for functions of several variables. In particular, provided the integrals exist, we have
\begin{align*}
 \widehat f(\xi,\eta)&= \frac1{2\pi}
    \iint f(x,y)\, e^{-i(\xi x + \eta y)}\,dx\,dy,\\
 \widecheck g(x,y)&= \frac1{2\pi} \iint g(\xi,\eta) e^{i(\xi x + \eta y)}\, d\xi\,d\eta.
\end{align*}

If $\mu$ is a finite measure on \R, its Fourier transform $\widehat \mu$ is an $\R\to\Co$ function:
\[
\widehat\mu(\zeta) = \frac1{\sqrt{2\pi}}\int_\R e^{-ix\zeta}d\mu(\zeta).
\]
Similarly, if $\nu$ is a finite measure on $\R^2$, its Fourier transform $\widehat \nu$ is an $\R^2\to\Co$ function:
\[
\widehat\nu(\xi,\eta)=
 \frac1{2\pi}\int_{\R^2}e^{-i(\xi x+\eta y)}\,d\nu(x,y).
\]

An \ltr\ function $\psi(x)$ is a \emph{Schwartz function} if it is infinitely differentiable and if it and its derivatives rapidly approach zero when $x\to\pm\infty$ in the sense that, for all nonnegative integers $j,k$, we have
\[ \lim_{x\to\pm\infty} \vp{x^j \frac{d^k\psi(x)}{dx^k}} = 0. \]
Schwartz functions are also known as \emph{test functions}.

The Fourier transform of a test function is a test function.
By the Plancherel theorem \cite[Theorem~A.19]{Hall}, the Fourier transform \F\ is a unitary operator on the test functions.
But these functions are dense in \ltr.
By continuity, \F\ is (or rather extends to) a unitary operator on the whole \ltr.

\subsection{Distributions, and exponential operators}
\label{sub:delta}

Dirac introduced a ``function'' $\delta(x)$ which is identically zero for all $x\ne0$ while $\delta(0)$ is infinite, so infinite that $\int_\R \delta(x)dx = 1$.
This makes no sense. $\delta$ is not a function in the usual sense. But it does make sense in the context of integrals of the form  $\int_\R f(x)\delta(x)dx$ which should be $f(0)$, at least for well behaved functions $f$.
Laurent Schwartz suggested viewing Dirac's $\delta$ and similar ``generalized functions'' as linear functionals on the space of test functions.
Thus, Dirac's $\delta$-function would be thought of as the linear functional on test functions $f$:
\begin{equation}\label{d}
 f\mapsto \int f(x) \delta(x) dx = f(0).
\end{equation}
If $0\ne c\in\R$, then
\begin{equation}\label{da}
\d(x) = \frac 1{|c|}\, \d\Big(\,\frac xc\,\Big).
\end{equation}
Indeed, if we use the substitution $y = x/c$ and keep integrating from $-\infty$ to $\infty$, then we have
\[
 \int f(x) \frac 1{|c|}\, \d\Big(\,\frac xc\Big) dx
= \int  f(cy) \delta(y) dy = f(0).
\]
Schwartz developed these ideas into a theory of \emph{distributions}, i.e.\ continuous linear real-valued functionals on the space of test functions (with the suitable topology).
Since then distributions play a major role in the theory of differential equations.

Some divergent integrals, e.g. $\int e^{itx}dt$, can be seen as distributions in that sense. In fact, as distributions,
\begin{equation}\label{e-d}
  \int e^{itx}dt = 2\pi\delta(x).
\end{equation}

Indeed,
\begin{align*}
\int dx\,f(x)\,\int e^{itx}dt
  &= \sqrt{2\pi} \int dt\,
     \frac1{\sqrt{2\pi}} \int f(x) e^{itx} dx\\
  &= \sqrt{2\pi} \int \widecheck f(t)\,dt\\
  &= 2\pi \cdot \frac1{\sqrt{2\pi}}
  \int \widecheck f(t)e^{-it0}dt = 2\pi f(0).
\end{align*}

The exponential $e^A$ of an operator $A$ on \ltr\ is the operator
\begin{equation}\label{exponent}
e^A = \sum_{k=0}^\infty \frac{A^k}{k!}
= I + A + \frac12 A^2 + \frac16 A^3 + \dots
\end{equation}
provided that series converges.

If $(X\psi)(x)= x\cdot\psi(x)$, then
\[
(e^X\psi)(x)
= \sum_{k=0}^\infty \frac1{k!} (X^k\psi)(x)
= \psi \cdot \sum_{k=0}^\infty \frac1{k!} x^k
= \psi \cdot e^x.
\]
If $D$ is the derivative operator $\frac d{dx}$ and $c$ a real number, then $e^{cD}\psi(x) = \psi(x+c)$.
Indeed,
\begin{align*}
 e^{cD} \psi(x)
&=\sum_{k=0}^\infty \frac{(cD)^k \psi(x)}{k!}
 = \sum_{k=0}^\infty \frac{D^kf(x)}{k!} c^k \\
&= \psi(x) + \frac{\psi'(x)}{1!}c
        + \frac{\psi''(x)}{2!}c^2
        + \frac{\psi'''(x)}{3!}c^3 + \dots
\end{align*}
which is the Taylor series of $\psi(x+c)$ around point $x$; think of $c$ as $\Delta x$.

\Qn I worry about convergence of the Taylor series.

\Ar The Taylor series certainly converges on analytic functions, in particular on Gaussian functions
\[
  \exp\left(-\frac{(x-a)^2}{2b^2}\right).
\]
The linear combinations of Gaussian functions are dense in $L^2(\R)$, and there is a unique continuous extension of $e^{cD}$ to $L^2$, namely the shift $f(x)\mapsto f(x+c)$.

\section{Pushforward measures, and uniqueness theorem}
\label{sec:pushforward}

We recall the definition of pushforward measures and then prove a measure-theoretic uniqueness theorem used in the proof of our main theorem in \S\ref{sec:char}.

Consider measurable spaces $M_1 = (\Omega_1,\Sigma_1)$ and $M_2 =
(\Omega_2,\Sigma_2)$. Let $\mu$ be a measure on $M_1$ and let a
function $f: \Sigma_1 \to \Sigma_2$ be measurable for $M_1, M_2$.

\begin{definition}[\S3.6 in \cite{Bogachev}]
The \emph{pushforward of $\mu$ along $f$}, a.k.a. the
\emph{$f$-pushforward of $\mu$} or the \emph{$f$-image} of $\mu$, is the measure
\[ f_*\mu(e) = \mu\big( f^{-1}(e) \big) \]
on $M_2$. \qef
\end{definition}

It is easy to check that $\nu = f_*\mu$ is indeed a measure on $M_2$.

\begin{proposition}\label{prp:image}
With notation as above, for every measurable function $g:\Omega_2\to\Co$ on $M_2$, if  $g(f(x))$ is $\mu$-integrable then $g$ is $\nu$-integrable and
\[ \int_{M_2} g(y) d\nu(y) = \int_{M_1} g(f(x)) d\mu(x). \]
\end{proposition}

For real-valued $f$, Proposition~\ref{prp:image} is a modification of Theorem~3.6.1 in book \cite{Bogachev} as described in the comments following the proof of theorem in the book. The generalization to \Co-valued functions is straightforward.

Recall that, by default, real Euclidean spaces $\R^k$ are equipped with the
Lebesgue measure.
Accordingly, measurable subsets of $\R^k$ are Lebesgue measurable, and
integrals are Lebesgue integrals.

\begin{theorem}[Uniqueness]\label{thm:unique}
  Let $\mu_1, \mu_2$ be finite Borel measures on $\R^2$.
If $(ax+by)_*\mu_1 = (ax+by)_*\mu_2$  for all $a,b$ not both zero,
then $\mu_1 = \mu_2$.
\end{theorem}

\begin{proof}
Let $a,b$ range over pairs of reals not both zero.
By Lemma~\ref{lem:diff}, $\mu = \mu_1 - \mu_2$ is a measure on $\R^2$. It suffices to prove that $\mu$ is the zero measure $\R^2$.
Let $\nu_{ab} =(ax+by)_*\mu$.

Every $\nu_{ab}$ is the zero measure on \R.
Indeed, if $s$ is a Borel subset $s$ of \R, let $S = \set{(x,y): ax+by\in s}$. We have
\begin{multline*}
(ax+by)_*\mu(s) = \mu(S) = \mu_1(S) - \mu_2(S) \\ =
(ax+by)_*\mu_1(s) - (ax+by)_*\mu_2(s) = 0.
\end{multline*}

We have.
\begin{align*}
  \widehat\nu_{ab}(\zeta)
  &=\int_{\R}e^{i\zeta t}\,d\nu_{ab}(t)
   =\int_{\R^2}e^{i\zeta(ax+by)}\,d\mu(x,y)\\
  &=\int_{\R^2}e^{i(a\zeta x+b\zeta y)}\,d\mu(x,y),
\end{align*}
where the second equality uses Proposition~\ref{prp:image} with $g(t) = e^{i\zeta t}$; the integrand $e^{i\zeta(ax+by)}$ is a bounded continuous function and therefore is integrable with respect to the (finite) measure $\mu$.

Comparing this with the Fourier transform of $\mu$,
\[
\widehat\mu(\xi,\eta)=\int_{\R^2}e^{i(\xi x+\eta y)}\,d\mu(x,y),
\]
we get
\[
\widehat\mu(a\zeta,b\zeta) = \widehat\nu_{ab}(\zeta).
\]
Since every $\nu_{ab}$ is the zero measure, every $\widehat\nu_{ab}(\zeta)=0$ for all $\zeta$. It follows that $\widehat\mu$ is the function zero.

By Proposition~3.8.6 in book \cite{Bogachev}, if two Borel measures on $\R^2$ have equal Fourier transforms, then they coincide. Applying this to $\mu$ and the zero measure, we conclude that $\mu$ is the zero measure.
\end{proof}

The theorem generalizes to higher dimensions, but we restrict our attention to $\R^2$.

\section{Marginal distributions}

Traditionally, for a (signed) probability distribution of several variables, the marginal (signed) distributions are defined only for single variables or subsets of the variables.
%
We extend this definition to linear functions of the variables, restricting attention to just two variables.

Let \P\ be a quasiprobability distribution on real plane $\R^2$ with coordinate axes $x$ and $y$, and let $(a,b)$ range over pairs of real numbers not both zero.

\begin{definition}
The \emph{$(ax+by)$ marginal} of \P\ is the pushforward measure (and in fact quasidistribution) $(ax+by)_*\P$ of \P\ along the function $z = ax+bp: \R^2\to\R$. \qef
\end{definition}

In particular,
$(x)_*\P(s) = \P\big( s\times\R \big)$, and
$(y)_*\P(s) = \P\big( \R\times s \big)$,
so that $(x)_*\P$ and $(y)_*\P$ are traditional marginals.

Now suppose that \P\ is given by a density function $f(x,y)$, so that $\P(s) = \iint_s f(x,y)\, dx\, dy$
for all measurable subsets $s$ of $\R^2$.
We show that in this case every marginal $(ax+by)_*\P$ is given by a density function which will be denoted $(ax+by)_*f$.

\begin{lemma}\label{lem:den}
For every pair $(a,b)$, the function
\[g(z) =
\begin{cases}
\displaystyle
\frac1{|b|} \int f\pa{x, \frac1b(z-ax)} dx
 &\mbox{if $b\ne0$}\\[10pt]
\displaystyle
\frac1{|a|} \int f\pa{\frac za, y} dy
 &\mbox{otherwise}
\end{cases}\]
is the density function $(ax+by)_*f$.
\end{lemma}

\begin{proof}
We consider the case $b\ne0$; the other case is similar (and a bit simpler).
It suffices to prove that $(ax+by)_*\P[u,v] = \int_u^v g(z)\,dz$ on intervals $[u,v]$ with $u\le v$. We have
\[ (ax+by)_*\P[u,v] = \iint_{u\le ax+by\le v} f(x,y)\,dx\,dy. \]
Let $z = ax+by$, so that $y = \frac1b(z - ax)$.
Change variables in the integral, from $x,y$ to $x,z$. The absolute value of the Jacobian determinant of this transformation is $\frac1{|b|}$, so we obtain
\begin{align*}
(ax+by)_*\P[u,v]
&= \iint_{u\le z\le v} \frac1{|b|} f\pa{x,\frac1b(z-ax)}\,dx\,dz\\
&= \int_u^v dz \int_\R \frac1{|b|} f\pa{x,\frac1b(z-ax)}\,dx
 = \int_u^v g(z)\,dz. \qedhere
\end{align*}
\end{proof}

\begin{lemma}\label{lem:constfactor1}
$(ax+by)_*\P(u,v) = (acx+bcy)_*\P(cu,cv)$ for every real $c\ne0$ and every open interval $(u,v)$ of \R.
\end{lemma}

\begin{proof}\mbox{}
\begin{align*}
(ax+by)_*\P(u,v)
&= \P\set{(x,y): u < ax + by < v}\\
&= \P\set{(x,y): cu < acx + bcy < cv} \\
&= (acx+cby)_*\P(cu,cv). \qedhere
\end{align*}
\end{proof}

\begin{lemma}[Lemma~6.4 in \cite{G245}]\label{lem:j2m}
For all real $a,b$ not both zero and every function $g:\R\to\R$, the following claims are equivalent.
\begin{enumerate}
\item $g$ is the density function $(ax+bp)_*f$.
\item $\displaystyle \widehat g(\zeta) =
  \sqrt{2\pi}\cdot\widehat f(a\zeta,b\zeta)$ where $\widehat f$ and $\widehat g$ are (forward) Fourier transforms of $f$ and $g$ respectively.
\end{enumerate}
\end{lemma}

The computation that  proves the lemma was essentially done in our
  proof of the uniqueness theorem above.

\section{Position, momentum, and their linear combinations}
\label{sec:wigner}

Consider one particle moving in one dimension.
A generalization to more particles in more dimensions is relatively straightforward.

In classical mechanics, the position $x$ and momentum $p$ of the particle determine its current state. The set of all possible classical states is the phase space of the particle.
In quantum mechanics, the state space of the particle is the Hilbert space $L^2(\R)$.

Using Dirac's bra-ket notation, we write \ket\psi\ for the vector given by function $\psi$.
Unit vectors \ket\psi\ represent states of the particle, and two unit vectors represent the same state if and only if they differ by a scalar factor $e^{i\theta}$.

\Qn How come a whole $\R \to\Co$ function is needed to represent just one quantum state?

\Ar Because, in quantum mechanics, a particle is also a wave. If it is in state \ket\psi, then $|\psi(x)|^2$ is the probability density at $x$ for finding the particle.
\braket\psi\psi\ is the total probability.
Accordingly, \braket\psi\psi\ must be 1, and this is why unit vectors are used to represent states.

In quantum mechanics, observable quantities are represented by Hermitian operators on the state space.
In particular, the position observable $X$ and momentum observable $P$ are (represented by) operators
\begin{align*}
 (X\psi)(x)& = x\cdot\psi(x), \\
 (P\psi)(x)& = -i\h \frac{d\psi}{dx}
\end{align*}
where $\h = h/2\pi$ is the reduced Planck constant; $h$ is the (unreduced) Planck constant.
We will simplify notation by assuming (by proper choice of
units) that $\h = 1$.

\Qn Functions $X\psi$ and $P\psi$ may fail to be square integrable.

\Ar Indeed, the operators $X,P$ are undefined in some states.

\Qn The formula for $P$ looks mysterious to me. Is momentum also related to the wave character of our particle?

\Ar Yes, it is. The momentum $p$ corresponds to the wavelength
$\lambda = h/p$ (de Broglie relation).
So, if a particle had an exact value $p$ of the momentum,
  its wave function would be (up to a scalar factor)
\[ \psi(x)=e^{2\pi i x/\lambda}=e^{ipx/\h} =  e^{ipx},\]
which is an eigenfunction of $P$ with eigenvalue $p$.
(Yes, this $\psi$ isn't in our Hilbert space. We'll return to this point in \S\ref{sec:dirac}.)
$P$ is designed to be the operator whose eigenvalues are momenta, just as $X$ is the operator whose eigenvalues (corresponding to ``eigenfunctions'' $\delta(x-q)$) are positions $q$.

There is a simple mathematical connection between the two operators: $\F P \F\mo = X$ where \F\ is the Fourier transform.
It suffices to verify this equality on test functions.
\begin{align*}
\F P(\psi)(\xi)
 &= \int P\psi(x)e^{-i\xi x} dx
 = \int -i \frac{d\psi(x)}{dx} e^{-i\xi x} dx\\
 &= i \int \psi(x) \frac{de^{-i\xi x}}{dx} dx
 = \xi \int \psi(x) e^{-i\xi x} dx = \\
 &= \xi\cdot\F(\psi)(\xi) =  X\F(\psi)(\xi).
\end{align*}
The third equality uses integration by parts;
the extra terms disappear because $e^{-i\xi x}$ is bounded and the test function $\psi$ approaches zero when the argument goes to $\pm\infty$.

\Qn The equality $\F P \F\mo = X$ makes me worry about the dimensions.
You mentioned earlier that, when $x$ is a length, as here, then the variable $\xi$ of the Fourier transform is a reciprocal length.
But here the variable of the Fourier transform seems to be a momentum.
How do you reconcile these dimensions?

\Ar
By convention, we're using units where $\h=1$, and the   dimension of $\h$ is momentum times length, so our convention   makes reciprocal length the same as momentum.
Thus, our $P$ is dimensionally correct.

\section{Dirac's kets}
\label{sec:dirac}

The spectral theory of self-adjoint operators in finite dimensional Hilbert spaces is relatively simple.
Suppose that \H\ is $n$-dimensional, and consider a self-adjoint operator $A$ on \H.
Since \H\ is self-adjoint, all its eigenvalues are real.
There exists an orthonormal basis \ket1, \dots, \ket{n} for \H\ composed of eigenvectors of $A$.
Let  $\l_1, \dots, \l_n$ be the corresponding eigenvalues.

For simplicity of exposition, we assume that all eigenvalues $\lambda_k$ are distinct (and thus non-degenerate since the number of them equals the dimension of \H). This suffices for our purposes in this paper, and it allows us to label the eigenvectors with the corresponding eigenvalues.
We may write \ket{\lambda_k} instead of \ket k.

It will be convenient, in the infinite dimensional case, to use an alternative characterization of the ``distinct eigenvalues'' assumption.
The operator $A$ is called \emph{cyclic} if there is a vector \ket\psi\ in \H\ such that the vectors
\[
\ket\psi,\, A\ket\psi,\, A^2\ket\psi,\, \dots,\, A^{n-1}\ket\psi
\]
span the whole space \H; such vector \ket\psi\ is \emph{$A$-cyclic}.

\begin{claim}
$A$ is cyclic if and only if the eigenvalues $\l_1, \dots, \l_n$ are all distinct.
\end{claim}

\begin{proof}
If the values are distinct, then $\sum_{k=1}^n \ket k$ is $A$-cyclic.
If, on the other hand, $\lambda_i = \lambda_j$, then for any vector $\ket\psi = \sum c_k \ket k$, all linear combinations of vectors $A^l \ket\psi$ have coefficients of \ket i and \ket j in the same ratio $c_i : c_j$. Thus these linear combinations fail to span \H.
\end{proof}

Think of \H\ as the state space of a quantum system where unit vectors represent states of the system (and two unit vectors represent the same state if and only if they are collinear).
Then every state \ket\psi\ can be written as
\begin{equation}\label{expansion}
\ket\psi = \sum_j \ket{\lambda_j}\,\braket{\lambda_j}\psi
\end{equation}
where scalars \braket{\lambda_j}{\psi} are probability amplitudes of the wave function \ket\psi, so that the corresponding probabilities are $\vp{\braket{\lambda_j}{\psi}}^2$.
Accordingly,
\[
A\ket\psi = \sum_j A\ket{\lambda_j}\,\braket{\lambda_j}\psi
             = \sum_j \l_j\ket{\lambda_j}\,\braket{\lambda_j}\psi
\]
and the expectation of $A$ in a state \ket\psi\ is
\[
\bra\psi A\ket\psi = \sum_j \l_j\braket{\psi}{\lambda_j}\,\braket{\lambda_j}\psi
                   = \sum_j \l_j \vp{\braket{\lambda_j}{\psi}}^2.
\]

The infinite dimensional case is much more involved.
Let $A$ be a self-adjoint operator on a infinite-dimensional space \H.
If you measure  $A$, you still receive some real number but it is not necessarily an eigenvalue. It is just an element of the spectrum
\[ \sigma(A) = \set{\lambda: \text{operator $A-\lambda I$ is not invertible}}
\]
of $A$.
Notice that the set of eigenvalues is
 \[
\set{\lambda: \text{operator $A-\lambda I$ is not one-to-one}}.
\]
In finite dimensions, any one-to-one linear operator is invertible,  but this principle fails in infinite dimensions.

As in the finite dimensional case, we make the simplifying assumption that the spectrum is simple in the sense that $A$ is \emph{cyclic}.
That is, there is a vector \ket\psi, called an \emph{$A$-cyclic} vector, such that the vectors
\[
\ket\psi,\, A\ket\psi,\, A^2\ket\psi,\, \dots
\]
are dense in whole Hilbert space \H.

How to generalize the spectral theory of self-adjoint operators in the finite dimensional case to the infinite dimensional case?
Paul Dirac came up with an elegant heuristic generalization \cite{Dirac} which works for operators like $aX+bP$ on \ltr.
We sketch Dirac's generalization, restricting attention to the Hilbert space \ltr\ and to cyclic self-adjoint operators $A$ on \ltr\ with $\sigma(A) = \R$.

For each spectrum value $r\in \R$, there is a generalized eigenvector, in short \emph{eigenket}, \ket{r} for $A$, so that $A\ket{r} = r\ket{r}$.
As in the finite dimensional case, we take advantage of the cyclicity assumption to label eigenkets \ket{r} by the corresponding spectrum value $r$.

\Qn What do you mean by eigenvector being generalized?

\Ar That it does not necessarily belong to \ltr.

\Qn That is confusing. Give me an example.

\Ar If $A$ is the momentum operator $P$,
then \ket{r} is the function $x\mapsto e^{irx}/\sqrt{2\pi}$. Indeed,
\[
P\pa{\frac{e^{irx}}{\sqrt{2\pi}}} = -i\,\frac d{dx}e^{irx} = r\pa{\frac{e^{irx}}{\sqrt{2\pi}}}.
\]

The finite-dimensional orthonormality requirement is replaced by \emph{\d-normality}:
\[ \braket{s}{r} = \delta(r-s). \]
In the case of $P$, using \eqref{e-d} we have
\[
\int \frac{e^{-isx}}{\sqrt{2\pi}}\cdot\frac{e^{irx}}{\sqrt{2\pi}}\,dx
= \frac1{2\pi}\int e^{i(r-s)x} dx = \d(r-s).
\]

The finite-dimensional expansion \eqref{expansion} with respect to the eigenvectors becomes the \emph{Dirac basis expansion}:
\[
\ket\psi = \int \ket{r}\,\braket{r}{\psi}\, dr
\]
where \braket{r}{\psi} is the density of probability amplitude, so that the corresponding probability density is
\begin{equation}\label{density}
D_A(r) = \vp{\braket{r}{\psi}}^2.
\end{equation}
The probability distribution on \R, given by the probability density
function $D_A(r)$, will be denoted $\P_A$.

We have
$  A\ket\psi = \int A\ket{r}\,\braket{r}{\psi}\, dr
              = \int r\ket{r}\,\braket{r}{\psi}\, dr $,
and the expectation of $A$ in state \ket\psi\ is
\begin{equation}\label{expectation}
\bra\psi A\ket\psi = \int r \vp{\braket{r}{\psi}}^2\, dr.
\end{equation}

In the case $A = P$, the density of probability amplitude is
\[
 \braket{r}{\psi}
 = \frac1{\sqrt{2\pi}} \braket{e^{irx}}{\psi(x)}
 = \frac1{\sqrt{2\pi}} \int e^{-irx}\psi(x)\, dx
 = \widehat\psi(r),
\]
the probability density is
\begin{equation}\label{pdensity}
D_P(r) = \vp{\widehat\psi(r)}^2,
\end{equation}
and $\P_P$ is the corresponding probability distribution.

\Qn Is there mathematical justification of Dirac's heuristic generalization?

\Ar The theory of rigged Hilbert spaces, see \cite{Bohm} for example, mathematically justifies the use of generalized eigenvectors%
\footnote{An alternative justification is provided by the spectral theory of operators \cite{Hall}.}.
The operators $A$ are subject to some constraints which are satisfied by the operators $X$, $P$ and their linear combinations.
See \S3 in \cite{Madrid} in this connection.

\Qn It this ``rigged'' as in rigged elections? What witty guy came up with this term?

\Ar This is a translation of a Russian term%
\footnote{\begin{otherlanguage}{russian} оснащенный; \end{otherlanguage} see \cite{GV}.}
meaning \emph{equipped} or \emph{rigged} as in ``rigging a ship for sailing.''

\Qn Will you tell me more about rigged Hilbert spaces?

\Ar Rigged Hilbert spaces deserve a separate column article, but the basic idea is that a Hilbert space \H\ is augmented with two additional spaces to obtain a so-called \emph{Gelfand triple}.

In our case, the triple is
\[ \Phi \subset \ltr \subset \Phi^\x \]
where $\Phi$ comprises test functions and $\Phi^\x$ comprises distributions.
Recall that we touched upon distributions in\S\ref{sub:delta}.

\Qn What kind of distributions are the eigenkets $\ket{r} = e^{irx}$ of $P$?

\Ar An antilinear functional
\[ \braket{f}{r} = \frac1{\sqrt{2\pi}} \int f^*(x) e^{irx} dx \]
on test functions $f$ \cite{Madrid2}.

\Qn I am confused. Earlier, in \S\ref{sub:delta}, you said that distributions are linear, not antilinear, functionals.

\Ar
Yes, the elements of $\Phi^\x$ are antilinear functionals serving as generalized kets while distributions are linear functionals serving as generalized bras.
In our case, it is safe to ignore the distinction and call both of them distributions.

\Qn Explain.

\Ar
To see what goes on, consider the finite dimensional case.
View elements of $\Co^n$ as column vectors, so the inner product \braket\phi\psi\ is given by the matrix product $\phi\dg\cdot\psi$.
Thus $\psi$ acts antilinearly on $\phi$ whereas $\phi$ acts linearly on $\psi$, but both are column vectors in the same space $\Co^n$. The same entities serve as linear and antilinear functionals.

From this point of view, an eigenket of an operator $A:\Co^n\to\Co^n$ for the eigenvalue $\lambda$ is a column vector $\psi$ such that $A\cdot\psi = \lambda \psi$.
An eigenbra is a column vector $\phi$ such that $\phi\dg \cdot A = \lambda \phi\dg$; equivalently, $A\dg\cdot\phi = \lambda^* \phi$.
When $A$ is self-adjoint and its eigenvalues $\lambda$ are therefore real, the eigenbras and eigenkets coincide.

The situation is similar in infinite dimensions, and the operators whose eigenkets and eigenbras we use are (essentially) self-adjoint.
Thus we can safely use the same entities as eigenkets and eigenbras.

\section{Linear combinations of position and momentum operators}

In this section, we apply Dirac's machinery to  linear combinations $aX+bP$ of the operators $X$ and $P$ where $a,b$ are real numbers not both zero.

One can argue that, from physical considerations, the spectrum of every $aX+bP$ is \R. This is supported by theory.
Every operator $aX + bP$ is self-adjoint%
\footnote{More exactly, $aX+bP$ is essentially self-adjoint \cite[Proposition~9.40]{Hall}, and every essentially self-adjoint operator has a unique self-adjoint extension \cite[Proposition~9.11]{Hall}.}.
The spectrum of every self-adjoint operator consists of reals \cite[Theorem~9.17]{Hall}.
Hence, for every $aX+bP$, the spectrum $\sigma(aX+bP) \subseteq \R$.
Furthermore, for every $aX+bP$, we will provide an eigenket \ket{r} of $aX+bP$ for every real $r$, so that $\sigma(aX+bP) = \R$.

\begin{lemma}\label{lem:constfactor2}
Let $Z = aX + bP$ where $a,b$ are real numbers not both zero, let $0\ne c\in\R$, and let $(u,v)$ be an interval in \R. Then, in every state \ket\psi,
\[ \Prob[Z\in (u,v)] = \Prob[cZ\in\set{cr: u<r<v}] =
\begin{cases}
  \Prob[cZ\in (cu,cv)] &\text{if }c>0,\\
  \Prob[cZ\in (cv,cu)] &\text{if }c<0.
\end{cases}
\]
\end{lemma}

\Qn This seems obvious. The two events are the same and so have the same probability.

\Ar $Z$ is an observable, and the probabilities are determined by the
rules of quantum mechanics.

\Qn But we can view $Z$ also as a random variable by repeatedly measuring it in a given state.
It is gratifying that both views give the same result, isn't it?

\Ar Agreed.

\begin{proof}
We consider the case $c>0$; the case $c<0$ is similar.

As we saw in the previous section, there is a \d-normal system \iset{f(r,x): r\in\R} where $f(r,x)$ is an eigenket of $Z$ for spectrum value $r$.
Then \iset{\frac1{\sqrt{c}} f(\frac rc,x): r\in\R} is a \d-normal system for $cZ$.
Indeed,
\[
(cZ)\left(\frac1{\sqrt{c}} f\big(\frac rc,x\big)\right)
  = \sqrt{c}\, Z f\big(\frac rc,x\big)
  = \sqrt{c}\, \frac r{c} f\big(\frac rc,x\big)
  =  r\cdot\frac1{\sqrt{c}} f\big(\frac rc,x\big),
\]
and, using the \d-normality of \iset{f(r,x): r\in\R} and using \eqref{da}, we have
\[
\Bigbraket{\frac1{\sqrt{c}} f\big(\frac sc,x\big)}
         {\frac1{\sqrt{c}} f\big(\frac rc,x\big)}
= \frac1c\d\Big(\frac sc - \frac rc\Big) = \d(r-s).
\]

By \eqref{density}, in a state \ket\psi, the probability density functions for $Z$ and $cZ$ are
\[
 \vp{\braket{f(r,x)}{\psi(x)}}^2 \quad\text{and}\quad
 \vp{\braket{f(r/c,x)}{\psi(x)}}^2
\]
Using the substitution $r=cs$, we have
\begin{align*}
\Prob[cZ \in (cu,cv)]
&= \int_{cu}^{cv} dr \left|\int_\R dx\, \frac1{\sqrt{c}} f^*(\frac rc,x) \psi(x)\right|^2\\
&= \int_u^v ds \vp{\int_\R dx\, f^*(s,x)\psi(x)}^2
= \Prob[Z \in (u,v)]. \qedhere
\end{align*}
\end{proof}

For every real $r$, the eigenket \ket{r} of $X$ for spectrum value $r$ is the delta function $\d(x-r)$, which acts on test functions according to
\[ f(x) \mapsto \braket{f}{r} = \int f^*(x)\d(x-r)\,dx = f^*(r).\]
$X\ket{r}$ is the distribution that sends a test function $f(x)$ to
\begin{align*}
& \int f^*(x)x\d(x-r)\,dx
 = \int f^*(x+r)(x+r)\d(x)\,dx \\
&= r f^*(r) = r\braket{f}{r}.
\end{align*}
\noindent
Thus, $X\ket{r} = r\ket{r}$.
In particular, $\sigma(X) = \R$.
Furthermore, these eigenkets form a \d-normal system:
\[\braket{s}{r} = \int \d(x-s) \d(x-r)\,dx = \d(r-s).\]
The density of probability amplitude is
$
 \braket{r}{\psi} = \int \d(x-r)\psi\,dx = \psi(r),
$
the probability density function is
\begin{equation}\label{xdensity}
D_X(r) = \vp{\psi(r)}^2.
\end{equation}
and $\P_X$ is the corresponding probability distribution.

Next we consider a case $Z = -cX + P$. It is still a special case, but the general case $Z = aX + bP$ easily reduces to that special case; we return to this issue in the next section.

The mapping defined by $U\ket\psi = e^{icx^2/2}\ket\psi$ is a unitary operator on the test functions:
\[
\braket{U\psi}{U\phi} = \bra{\psi}U\dg U\ket{\phi} =
\braket{\psi}{\phi}.
\]
We have $UPU\mo = Z$ on test functions \cite[\S6]{G245}. Indeed,
\begin{align*}
\pa{UPU^{-1}}\psi &= e^{icx^2/2}\cdot
\pa{ -i \frac d{dx} \big(e^{-icx^2/2} \psi \big) }
 = -c x\psi - i \frac{d\psi}{dx} \\
&= -c X\psi + P\psi   = Z\psi.
\end{align*}
Accordingly, one may expect that $U$ transforms generalized eigenvectors of $P$ into those of $Z$.
This intuition happens to be correct.

For each real $r$, the distribution $e^{irx + icx^2/2}$ is a generalized eigenvector of $Z$ for $r$:
\begin{align*}
Z e^{irx + icx^2/2}
&= (-c X + P) e^{irx + icx^2/2}
 =-c x e^{irx+icx^2/2} - i \frac d{dx}e^{irx+icx^2/2} \\
&= e^{irx + icx^2/2} \big[ -c x - i(ir + icx) \big]
 = r e^{irx + icx^2/2}.
\end{align*}
It follows that the spectrum of $Z$ is \R.

For each real $r$, let \ket{r} be the eigenket $\frac1{\sqrt{2\pi}}e^{irx + icx^2/2}$ of $Z$ for generalized eigenvalue $r$. These eigenkets form a \d-normal system:
\[
\int \frac{e^{-isx-icx^2/2}}{\sqrt{2\pi}} \cdot \frac{e^{irx+icx^2/2}}{\sqrt{2\pi}}\,dx
= \frac1{2\pi}\int e^{i(r-s)x} dx = \d(r-s).
\]
Accordingly, the density of probability amplitude is
\[
 \braket{r}{\psi}
 = \frac1{\sqrt{2\pi}}\braket{e^{irx+icx^2/2}}{\psi(x)}
 = \frac1{\sqrt{2\pi}} \int e^{-irx+icx^2/2}\psi(x)\, dx,
\]
the probability density is
\begin{equation}\label{zdensity}
 D_Z(r) = \frac1{2\pi} \vp{\braket{e^{irx+icx^2/2}}{\psi(x)}}^2,
\end{equation}
and the corresponding probability distribution is $\P_Z$.

\section{Wigner's quasiprobability distribution}

In every \ltr\ state \ket\psi, Wigner's quasiprobability $\P_W$ is given by probability density function
\begin{equation}\label{denw}
w(x,p) = \frac1{2\pi} \int_\R \psi^*(x+\frac{\gamma\h}2)
\psi(x-\frac{\gamma\h}2) e^{i\gamma p}\,d\gamma.
\end{equation}
The integral converges in every \ltr\ state \ket\psi.
According to Wigner, the $x$-marginal and $p$-marginal of his distribution are the probability distributions for $X$ and $P$ respectively \cite{Wigner}.

\Qn You don't take advantage of using units where $\h=1$. I guess it doesn't hurt to keep \h\ in this case.

\Ar Exactly.

\Qn Is this obvious that the $x$-marginal and $p$-marginal of Wigner's distribution are the probability distributions for $X$ and $P$\,?

\Ar It is certainly easier to verify than the claim that an arbitrary $(ax+by)$-marginal of Wigner's distribution is the probability distribution for $aX+bY$.

\begin{lemma}\label{lem:w1}
In every state \ket\psi, $(x)_*\P_W = \P_X$.
\end{lemma}

\begin{proof}
It suffices to prove that, in every state \ket\psi, the $x$-marginal $x_*w$ of Wigner's density is the density function $\vp{\psi}^2$. Since two density function coincide if they are proportional, we may neglect constant factors.

By Lemma~\ref{lem:den}, up to constant factors, the $x_*w$ density function is
\begin{align*}
& x_*w = \int w(x,p)dp = &&\text{by Lemma~\ref{lem:den}} \\
& \int\left[\int \psi^*(x+\frac{\gamma\h}2)\,
 \psi(x-\frac{\gamma\h}2)\, e^{i\gamma p}\,d\gamma\right]dp =\\
& \int \left[ \int e^{i\gamma p} dp \right]
 \psi^*(x+\frac{\gamma\h}2)\,  \psi(x-\frac{\gamma\h}2)\,d\gamma =&&\text{by \eqref{e-d} } \\
& \int \delta(\gamma)
 \psi^*(x+\frac{\gamma\h}2)\,  \psi(x-\frac{\gamma\h}2)\,d\gamma = &&\text{by \eqref{d}} \\
& \psi^*(x)\,\psi(x) = \vp{\psi}^2 \qedhere
\end{align*}
\end{proof}

\begin{lemma}\label{lem:w2}
$\displaystyle
w(x,p)=\int\widehat\psi(p+\frac\gamma2)^*\,
\widehat\psi(p-\frac\gamma2)\, e^{-i\gamma x}\,d\gamma$\\
up to a constant factor.
\end{lemma}

\begin{proof}
Ignoring constant factors and using the formula
  \[
\psi(x)=\int\widehat\psi(\xi)\,e^{i\xi x}\,d\xi
\]
 for the inverse Fourier transform, we get
\begin{align*}
w(x,p)
&= \int\psi(x+\frac\gamma2)^*\,\psi(x-\frac\gamma2)\,
 e^{i\gamma p}\,d\gamma \\
&=\iiint\widehat\psi(\xi)^*\,
  e^{-i\xi(x+\frac\gamma2)}\,\widehat\psi(\eta)\,
  e^{i\eta(x-\frac\gamma2)}\,e^{i\gamma p}\,d\xi\,d\eta\,d\gamma \\
&= \iint\left[e^{-i\gamma(\frac\xi2+\frac\eta2-p)}d\gamma \right]
\widehat\psi(\xi)^*\,\widehat\psi(\eta)\,
e^{-ix(\xi-\eta)}\, d\xi\,d\eta \\
&= \iint\widehat\psi(\xi)^*\,\widehat\psi(\eta)\,
 e^{-ix(\xi-\eta)}\,
\delta(\frac\xi2 +\frac\eta2-p)\,d\xi\,d\eta.
\end{align*}
The $\delta$ function makes it easy to perform the integration with
respect to $\eta$.
Just substitute $2p-\xi$ for $\eta$ in the remaining factors.
\[
w(x,p)=\int\widehat\psi(\xi)^*\,\widehat\psi(2p-\xi)\,
 e^{-ix(2\xi-2p)}\,d\xi.
\]
Change variables in the integral to $\gamma=2\xi-2p$, so $\xi$ becomes $p+\frac\gamma2$.
\[w(x,p)=\int\widehat\psi(p+\frac\gamma2)^*\,
\widehat\psi(p-\frac\gamma2)\,e^{-i\gamma x}\,d\gamma. \qedhere\]
\end{proof}

\begin{corollary}
The $(p)$-marginal of Wigner's distribution is the probability distribution $\P_P$ for $P$.
\end{corollary}

The proof is similar to that of Lemma~\ref{lem:w1}, except the formula of Lemma~\ref{lem:w2} is used.

The generalization to arbitrary $(ax+bp)$-marginals will be proved in \S\ref{sec:char}.
A key role in that proof is played by the following lemma.

\begin{lemma}[Lemma~6.6 in \cite{G245}]\label{lem:key}
For all test functions $\psi$ and real numbers $\alpha,\beta$, not both zero,
\[
 \bra\psi e^{-i (\alpha X + \beta P)}\ket\psi =
 e^{i\alpha\beta\hbar/2}
 \int\psi^*(y)e^{-i\alpha y}\psi(y-\beta\hbar)\,dy.
\]
\end{lemma}

Actually, Lemma~6.6 in \cite{G245} speaks about states \ket\psi\ that are ``nice'' in the sense that the function $\psi$ is smooth and compactly supported.
But the lemma and its proof obviously remain valid for test functions $\psi$.

\section{Characterization theorem}
\label{sec:char}

\begin{theorem}\label{thm:main}
In every state \ket\psi\ in \ltr, Wigner's quasidistribution $\P_W$ is the unique quasidistribution on $\R^2$ such that, for all real numbers $a,b$ not both zero, the marginal $(ax+bp)_*\P_W$ is correct in the sense that it coincides with the probability distribution $\P_Z$ of the observable $Z = aX + bP$.
\end{theorem}

\begin{proof}
The uniqueness follows from Theorem~\ref{thm:unique}.

We need to prove the equality $(ax+bp)_*\P_W = \P_Z$ for every pair of reals $a,b$ not both zero and in every state \ket\psi.
By virtue of Lemmas~\ref{lem:constfactor1} and \ref{lem:constfactor2}, we restrict attention to two cases:
\begin{enumerate}
\item $a=1$ and $b=0$, so that $Z = X$,
\item $b=1$, so that $Z = aX + P$.
\end{enumerate}
Case~(1) is taken care of by Lemma~\ref{lem:w1}.
In the rest of the proof we consider case~(2).
Even though $b=1$, we sometimes write $aX+bP$ and $ax+bp$ anyway.

Fix an arbitrary real number $a$.
We prove that, in every state,
\begin{equation}\label{pcorrect}
(ax+bp)_*\P_W = (ax+p)_*\P_W = \P_Z.
\end{equation}

In every state \ket\psi, the quasidistribution $\P_W$ is given by the probability density function $w(x,p)$ specified by formula \eqref{denw}. By Lemma~\ref{lem:den}, the marginal quasidistribution $(ax+p)_*\P_W$ is given by quasiprobability density function
\begin{align*}
g(r) &= \int w(x,r-ax)\,dx \\
&=  \frac1{2\pi} \iint  \psi^*(x+\frac{\gamma\h}2)\,
\psi(x-\frac{\gamma\h}2)\, e^{i\gamma(r-ax)}\,d\gamma\,dx.
\end{align*}
The probability distribution $\P_Z$ is given by the probability density function $D_Z$ of $Z$ specified in \eqref{zdensity}.
To prove \eqref{pcorrect}, it suffices to prove
\begin{equation}\label{dcorrect}
 g(r) = D_Z(r).
\end{equation}

Both sides of \eqref{dcorrect} are continuous as functions of the state in the \ltr\ metric.
By continuity, it suffices to prove that the equality \eqref{dcorrect} holds in every ``nice'' state \ket\psi\ provided that the nice states are dense in \ltr.

Theorem~6.2 in \cite{G245} does just that. In that theorem, a state \ket\psi\ is nice if $\psi$ is smooth and compactly supported.
While Theorem~6.2 addresses both, the uniqueness and the correctness aspects,
the emphasis in its proof is on uniqueness, and the correctness proof may be a bit confusing.
We explain it here.
Notice that two density functions coincide if they are proportional.
Accordingly, we ignore constant factors in equations below.

Since (ignoring the factor $1/2\pi$)
\[
w(x,p) = \int \psi^*(x+\frac{\gamma\h}2)\, \psi(x-\frac{\gamma\h}2)\, e^{i\gamma p}\,d\gamma,
\]
its Fourier transform is
\[
  \widehat w(\alpha,\beta)=\iiint dx\,dp\,d\gamma\,
  \psi^*(x+\frac{\gamma\h}2)\,
\psi(x-\frac{\gamma\h}2)\,
 e^{i\gamma p}\, e^{-i\alpha x}\,e^{-i\beta p}.
\]
Here $p$ occurs only in two of the exponential factors, so the
integration over $p$ produces
\[
\int dp \,e^{i(\gamma-\beta)p}=\delta(\gamma-\beta).
\]
The delta function now makes the integral over $\gamma$ trivial; just substitute $\beta$ for $\gamma$ in the integrand. Thus,
\[
\widehat w(\alpha,\beta)=\int dx\, \psi^*(x+\frac{\beta\h}2)\,
\psi(x-\frac{\beta\h}2)\,e^{-i\alpha x}.
\]
Introducing a new integration variable $y=x+\frac{\beta\h}2$, we
get
\[
  \widehat w(\alpha,\beta)=\int dy\,\psi^*(y)\,\psi(y-\beta\h)\,e^{-i\alpha y}\,
  e^{i\alpha\beta\h/2}.
\]
By Lemma~\ref{lem:key},
\[
  \widehat w(\alpha,\beta)=\bra\psi e^{-i(\alpha X+\beta P)}\ket\psi.
\]
In particular, if $Z$ is defined as $aX+bP$ and if we substitute
$a \zeta$ and $b \zeta$ for $\alpha$ and $\beta$, we get
\[
\widehat w(a\zeta,b\zeta)=\bra\psi e^{-i\zeta Z}\ket\psi.
\]
By Lemma~\ref{lem:j2m}, $\bra\psi e^{-i\zeta Z}\ket\psi$ is the Fourier transform of the marginal $g(z) = (ax+bp)_*w(x,p)$.
But this same $\bra\psi e^{-i\zeta Z}\ket\psi$ is also, up to a constant factor, the Fourier transform of the density function $D_Z$ of $Z$ in the state $\psi$.
Indeed, by \eqref{exponent},
\[
e^{-i\zeta Z} = \sum_k \frac1{k!}(-i\zeta Z)^k,
\]
and we restrict attention to test functions $\psi$, so that convergence is no problem. We have
\[
\bra\psi e^{-i\zeta Z}\ket\psi
= \Big\langle\psi\,\Big|\,\sum_k\frac1{k!}(-i\zeta Z)^k\,\Big|\psi\Big\rangle
= \sum_k \frac{(-i\zeta)^k}{k!} \bra\psi Z^k \ket\psi
\]
Now we use Dirac's machinery. Let \ket{r} be the eigenket $\frac1{\sqrt{2\pi}}e^{irx + icx^2/2}$ of $Z$ for spectrum value $r$. We have $Z^2\ket r = Z(r\ket r) = r^2\ket r$ and similarly for other powers of $Z$. By the preceding computation, \eqref{expectation}, and \eqref{zdensity}, we have
\begin{align*}
\bra\psi e^{-i\zeta Z}\ket\psi
&= \sum_k \frac{(-i\zeta)^k}{k!} \int r^k |\braket r\psi|^2\,dr
= \int \sum_k \frac{(-i\zeta)^k}{k!} r^k |\braket r\psi|^2\,dr \\
&= \int e^{-i\zeta r} |\braket r\psi|^2\,dr
= \int e^{-i\zeta r} D_Z(r)\,dr,
\end{align*}
as required.
\end{proof}

\subsection*{Acknowledgment}
We thank Rafael de la Madrid for useful comments.


\begin{thebibliography}{99}

\bibitem{Bertrand} Jacqueline Bertrand and Pierre Bertrand, ``A tomographic approach to Wigner's function," Foundations of Physics 17:4 (1987) 397--405

\bibitem{G224} Andreas Blass and Yuri Gurevich, ``Negative probabilities," Bulletin of EATCS, February 2014


\bibitem{G245} Andreas Blass and Yuri Gurevich, ``Negative probabilities: What they are for," Journal of Physics A 54 (2021), article 315303,
    \url{https://doi.org/10.1088/1751-8121/abef4d}

\bibitem{Bohm} Arno B\"ohm, ``The rigged Hilbert space and quantum mechanics," Springer 1978

\bibitem{Bogachev} Vladimir I. Bogachev, ``Measure theory,'' Vol.~1, Springer 2007

\bibitem{Dirac} Paul Dirac, ``The principles of quantum mechanics," 4th edition, Oxford University Press 1958


\bibitem{Feynman} Richard P. Feynman, ``Simulating physics with computers," International Journal of Theoretical Physics, Vol. 21, Nos. 6/7 (1982) 467--488

\bibitem{GV} Israel M. Gelfand and Naum Y. Vilenkin, ``Generalized
  functions," Vol.~4, Academic Press 1964; the Russian original
  published in 1961 by \begin{otherlanguage}{russian}
    Физматлит \end{otherlanguage}

\bibitem{Hall} Brian C. Hall, ``Quantum theory for
    mathematicians," Springer 2013

\bibitem{Madrid} Rafael de la Madrid, ``Quantum mechanics in rigged
  Hilbert space language," PhD thesis, Universidad de Valladolid, 2001

\bibitem{Madrid2} Rafael de la Madrid, ``The role of the rigged Hilbert space in quantum mechanics," European Journal of Physics 26 (2005) 287--312, doi:10.1088/0143-0807/26/2/008


\bibitem{Rudin} Walter Rudin, ``Principles of mathematical analysis," McGraw-Hill, 3rd edition, 1976

%


\bibitem{Wigner} Eugene P.~Wigner, ``On the quantum correction for thermodynamic equilibrium," Physical Review 40 (1932) 749--759




\end{thebibliography}
\end{document}